\documentclass{llncs}
\usepackage{makeidx}

\newcommand{\commentout}[1]{}

\newcommand{\QED}{\hbox{\hskip 4pt \vrule width 5pt height 6pt depth
    1.5pt\hskip 2pt}}

\newtheorem{thm}{Theorem}
\newtheorem{cor}{Corollary}
\newtheorem{lem}{Lemma}
\newtheorem{obsr}{Observation}

\begin{document}
\date{}
\title{Longest Common Subsequence in $k$ Length Substrings\thanks{A preliminary version of this paper appeared in SISAP 2013.}}

\author{
Gary Benson\inst{1} \thanks{This work was partially funded by NSF
grant IIS-1017621.}, Avivit Levy\inst{2} and B.~Riva
Shalom\inst{2} } \institute{ Department of Computer Science,
Boston University, Boston,  MA 02215 , {\tt gbenson@bu.edu};
  \and  Department of Software Engineering,
Shenkar College,  Ramat-Gan 52526, Israel; {\tt $\{$avivitlevy,
rivash$\}$@shenkar.ac.il} }

\maketitle

\begin{abstract}
In this paper we define a new problem, motivated by computational
biology, $LCSk$ aiming at finding the  maximal number of $k$
length $substrings$, matching in both input strings while
preserving their order of appearance. The traditional LCS
definition is a special case of our problem, where $k = 1$. We
provide an algorithm, solving the general case in $O(n^2)$ time,
where $n$ is the length of the input strings, equaling the time
required for the special case of $k=1$. The space requirement of
the algorithm is $O(kn)$. 
 We also define a complementary $EDk$ distance measure and show that
$EDk(A,B)$ can be computed in $O(nm)$ time and $O(km)$ space,
where $m$, $n$ are the lengths of the input sequences $A$ and $B$
respectively.
\end{abstract}


\section{Introduction}\label{s:intro}
The {\em Longest Common Subsequence} problem, whose first famous
dynamic programming solution appeared in 1974~\cite{WF-74}, is one
of the classical problems in computer science. The widely known
string version appears in Definition \ref{d:lcs1}.

\begin{definition}\label{d:lcs1}
  {\em The String Longest Common Subsequence ($LCS$) Problem}:\\
\begin{tabular}{ll}
 Input: &  Two sequences $A = a_1a_2\ldots a_n,$  $B=b_1b_2\ldots b_n$ over alphabet
$\Sigma$. \\
  Output: & The length of the longest subsequence common of both
strings, \\ &where a subsequence is a sequence that can be derived
from \\ &another sequence by deleting  some elements without
changing
\\ & the order of the remaining elements.\\
\end{tabular}
\end{definition}
For example, for the sequences appearing in Figure \ref{f:pair
example}, $LCS(A, B)$ is 5, where a possible such subsequence is
$T\ T\ G\ T\ G$.

\paragraph{\textbf{Remark:}} We alternately use the terms
 \emph{string} and \emph{sequence} throughout the paper, since both terms are
 common as the input of the LCS problem. Nevertheless, a
 \emph{subsequence} is obtained from a sequence by deleting symbols
at any index we want,   while a \emph{substring} is a consecutive
part of the string.

The LCS problem is motivated by the need to
measure similarity of sequences and strings, and it has been very well studied (for a survey,
see~\cite{bhr:00}). The well known dynamic programming solution
\cite{Hi-75} requires running time of $O(n^2)$, for two input
strings of length $n$.

The LCS problem has also been investigated
on more general structures such as trees and matrices
\cite{ahkst:06}, run-length encoded strings \cite{als:99},
weighted sequences \cite{ags:10}, \cite{bjv:12} and more. Many
variants of the LCS problem were studied as well, among which LCS
alignment \cite{lz:01}, \cite{lsz:03}, \cite{lmz:04}, constrained
LCS \cite{T:03}, \cite{cc:11}, restricted LCS \cite{ghll:10} and LCS approximation \cite{lln:11}.

\paragraph{\textbf{Motivation.}} The LCS has been also used as a measure of
sequence similarity for biological sequence comparison. Its
strength lies in its simplicity which has allowed development of
an extremely fast, bit-parallel computation which uses the bits in
a computer word to represent adjacent cells a row of the LCS
scoring matrix and computer logic operations to calculate the
scores from one row to the next \cite{ad:86}, \cite{cipr:01},
\cite{h:04}. For example, in a recent experiment, 25,000,000
bit-parallel LCS computations (sequence length = 63) took
approximately 7 seconds on a typical desktop computer
\cite{bhl:13} or about 60 times faster than a standard algorithm.
This speed makes the LCS attractive for sequence comparison
performed on high-sequencing data. The disadvantage of the LCS is
that it is a crude measure of similarity because consecutive
matching letters in the LCS can have different spacings in the two
sequences, i.e., there is no penalty for insertion and deletion. Indeed,
 as is well-known, there is a strong relation between the total number of
 insertions and deletions and the LCS. However, there are no limitations on
 the ``distribution'' of such insertions and deletions. Consider for example
 the following two sequences:
\[A = ( G T G)^{n/3}\]
\[ B = (T C C )^{n/3}\]
In these sequences the LCS is quite large, of size $n/3$, but
no two matched symbols are consecutive in both sequences.
Any common subsequence of this size ``put together'' separated
elements implying a rather ``artificial'' similarity of the
sequences. \commentout{
\[A=1^{n/4}0^{n/2}1^{n/4}\]
\[ B=0^{n/4}1^{n/2}0^{n/4}\]
In these sequences the LCS is quite large, of size $n/2$, but any
common subsequence of this size ``put together'' elements that are
separated by at least $n/4$ positions. }

What is proposed here is
a definition of LCS that makes the measure of similarity more
accurate because it enables forcing adjacent letters in the LCS to
be adjacent in both sequences. In our problem, the common
subsequence is required to consist of k length substrings. A
formal definition appears in Definition~\ref{d:LCSk}.

\begin{definition}\label{d:LCSk}
   \emph{The Longest Common Subsequence in k Length Substrings
    Problem ($LCSk$)}:\\
\begin{tabular}{ll}
 Input: &  Two  sequences $A= a_1a_2 \ldots a_n$,  $B=b_1b_2 \ldots b_n $  over alphabet
$\Sigma$.\\
   Output: & The maximal $\ell$ s.t. there are $\ell$ substrings,\\ \
    &  $a_{i_1}...a_{i_{1+k-1}}\ldots
a_{i_\ell}...a_{i_\ell +k-1}$, identical to
$b_{j_1}b_...{j_{1+k-1}}\ldots
b_{j_\ell}...b_{j_\ell +k-1}$ \\
& where $\{ a_{i_f} \}$ and $\{ b_{j_f}\}$ are in increasing order for $1\leq f \leq \ell $   and\\
&    where two $k$ length substrings in the same sequence, do not overlap.\\
\end{tabular}
\end{definition}

We demonstrate $LCSk$ considering the sequences appearing in
Figure \ref{f:pair example}.
\begin{figure}[h]
\begin{center}
A =  \begin{tabular}{llllllll}

       \footnotesize{1} & \footnotesize{2}  & \footnotesize{3} & \footnotesize{4} &\footnotesize{5} & \footnotesize{6} & \footnotesize{7} &\footnotesize{8}    \\
      T$ \quad $    & G$ \quad $   & C$ \quad $   & \textbf{G}$ \quad $    & \textbf{T}$ \quad $     &  G$ \quad $     & \textbf{T}$ \quad $    & \textbf{G}\\
\end{tabular}\\

B =  \begin{tabular}{llllllll} 
      \textbf{G}$ \quad $ & \textbf{T}$ \quad $ & T$ \quad $    & G$ \quad $      & \textbf{T}$ \quad $    &  \textbf{G}$ \quad $      & C$ \quad $   & C\\ 
\end{tabular}
\normalsize \caption{An LCS2 example }\label{f:pair example}
\end{center}
\end{figure}

A possible common subsequence in pairs $(k=2)$ is $G\ T\ T\ G$
obtained from $a_4,$ $a_5,$ $a_7,$ $a_8$ and $b_1,b_2,b_5,b_6$.
Though $a_6 = b_4$, and such a match preserves the order of the
common subsequence, it cannot be added to the common subsequence
in pairs, since it is a match of a single symbol. For $k=3$, one
of the  possible solutions is $T\ G\ C$ achieved  by matching
$a_1,$ $a_2,$ $a_3,$ with $b_5,b_6,b_7$. For $k = 4$ a possible
solution is $T\ G\ T\ G$ obtained from matching $a_5,$ $a_6,$
$a_7,$ $a_8$ and $b_3,b_4,b_5,b_6$. Note that in the last two
cases the solution contains a single triple and a single quadruplet.

The paper is organized as follows: Section~\ref{s:preliminaries}
gives some preliminaries. The solution for the $LCSk$ problem is
detailed in Section~\ref{s:algorithm}. In Section~\ref{s:editdistance} we refer to the complementary edit distance measure, $EDk$. Section~\ref{s:conclusion} concludes the paper.

\section{Preliminaries}\label{s:preliminaries}
 The LCSk problem is a generalization of the LCS problem. We might
consider using the solution of the latter in order to solve the
former. If we perform the LCS algorithm on the input sequences, we
can backtrack the dynamic programming table and mark the symbols
participating in the common subsequence. We can then check whether
those symbols appear in consecutive $k$ length substrings in both
input sequences, and delete them if not. Such a procedure
guarantees a common subsequence in $k$ length substrings but not
necessarily the optimal length of the common subsequence. For
example consider $LCS2$ of the sequences appearing on Figure
\ref{f:pair example}. Applying the LCS algorithm on these strings
may yield $T\ T\ G\ T\ G$, containing a single non-overlapping
pair matching while there exists LCS2 of  $T\ G\ T\ G$ having two
pair matchings. Hence, a special algorithm designed for
\emph{LCSk} is required.

As the LCSk problem considers matchings of $k$
consecutive symbols, throughout this
paper we call such a matching a $k\ matching$. We will also need the term \emph{predecessor}. We use the following definitions for these terms:\\

\begin{definition}\label{d:kmatch}
\begin{displaymath} kMatch(i,j) =
\left\{ \begin{array}{ll}
1 & if\ a_{i+f}= b_{j+f},\ for\ every \ 0 \leq f \leq k-1 \\
0 &  Otherwise\\
\end{array} \right.
\end{displaymath}
If $kMatch(i,j) = 1$, the matching substring is denoted (i,j).
\end{definition}

\begin{definition}\textbf{Predecessors.}
 Let \textit{$ candidates(i,j)$} be the set of all longest common subsequences, consisting of
 k matchings, of  prefix  $A[1...i+k-1]$ and  prefix  $B[1...j+k-1]$. Then let
\textit{$ pred(i,j)$} be  all the possible last  $k$ matchings in
 $candidates(i,j)$. That is, $pred(i,j) = \{ (s,t)| \exists c \in   candidates(i,j),\ where\ $ $(s,t) \ is\
 the\  last\  $ $  k\ matching\  in\ c\}$.

We define the \textit{length} of $p \in pred(i,j)$ derived from candidate $c$, to be the number of $k$ matchings in $c$ and denote it by $|p|$.
\end{definition}

\paragraph{\textbf{Example.}} Consider LCS2 of the sequences of Figure
\ref{f:pair example}. Let $candidates(5,3)$ be the common
subsequences in pairs of $B[1...4] = G\ T\ T \ G$  and of
$A[1...6]=T \ G\ C \    G  \    T\ G $, thus, $candidates(6,4)$
contains$ \{ TG, GT \} $. $TG$ can be obtained in two ways:
$a_1a_2 $ matched to $b_3b_4$, or $a_5a_6$ matched to $b_3b_4$,
and  $GT$ by $a_4a_5$ matched to $ b_1b_2$ therefore, we have
$pred(i,j) = \{ (1,3), (5,3), (4,1)\}$. In this example all
predecessors are of length 1. Keeping the predecessors enables
backtracking to reveal the longest common subsequence in $k$
length substrings itself.\\

The following Lemma is necessary for the correctness of the
solution.
\begin{lem}\label{l:length}
Let $p_1$,$p_2 \in pred(i,j) $, then if
 $|p_1| + 1 = |p_2|$,
 then any maximal common subsequence  of $k$ length substrings using the $k$ matching $p_2$ has
 length greater or equal to that using the $k$ matching $p_1$.
\end{lem}
\begin{proof}
Suppose $p_1= (s,t)$ and $p_2 = (s',t')$. Several cases are
possible for $p_1, p_2$:
\begin{enumerate}
 \item If $s' < s$ and $ t' <
t$, then the candidate whose last $k$ matching is  $p_2$ might be
further extended till $A[s]$ and $B[t]$, enlarging the difference
between $p_1$ and $p_2$.
 \item If  $s' = s $ and $t' = t$ both predecessors have the same
 opportunities for extension.
 \item If $s + k -1 < s'$ and $ t + k - 1
< t'$,  the $k$ matching $(s',t')$ can be added to the candidate
whose last $k$ matching is $(s,t)$, contradicting its maximality.
 \item  If there is an overlap
between the $k$ matchings represented by the  predecessors, $s <
s' < s+ k$ or $ t < t' < t + k$,  starting from $a_{s'+k}$, every
$k$ matching can be used to extend the common subsequence in $k$
length substring, represented by both predecessors. However, the
subsequence using $p_1$ cannot have an additional $k$ matching
before $a_{s'+k}$, as overlaps are forbidden. Consequently, the
difference between the length of $p_1$ and $p_2$ is preserved in
the extended maximal common subsequences.
\end{enumerate} \QED
\end{proof}

\section{Solving the LCSk Problem}\label{s:algorithm}
 As in other LCS variants, we solve the problem using a dynamic
programming algorithm. We denote by $LCSk_{i,j}$  the longest
common subsequence, consisting of $k$ matchings in the prefixes
$A[1...i+k-1]$ and $B[1...j+k-1]$. Lemma \ref{l:kfilling} below,
formally describes
 the computation of $LCSk_{i,j}$.

\begin{lem}{\textbf{The LCSk Recursive Rule.}}\label{l:kfilling}
\begin{displaymath} LCSK_{i,j}=
max \left\{ \begin{array}{l}
      LCSk_{i,j-1}, \\
      LCSk_{i-1,j}, \\
   LCSk_{i-k,j-k} + kMatch(i,j) \\
   \end{array}
  \right.
\end{displaymath}
\end{lem}
\begin{proof}
$LCSK_{i,j}$ contains the maximal number of common $k$
length substrings, preserving their order in the input sequences.
A possible subsequence can be constructed by matching the
substrings $a_i,\ldots , a_{i+k-1}$ with $b_j,\ldots , b_{j+k-1}$,
in case all $a_{i+f}$ and  all $b_{j+f}$, for $0\leq f\leq k-1$,
are not part of previous  $k\ matching$s. This is guaranteed
when considering the prefixes $A[1..i-k]$ and $B[1..j-k]$ while
trying to  extend by one the common subsequence for cell
$LCSk_{i,j}$. Another option of extending the subsequence is by
using the $k$ matching $(s,j)$ , for $s < i$. Similarly, we can
use the k matching $(i,t)$  for $t < j$. Note that the options of
extending $LCSk_{i-f, j-f}$, for $1\leq f\leq k-1$ is included in
both $ LCSk_{i,j-1}$ and $ LCSk_{i-1,j}$. These claims can be
easily proven using induction. \QED
\end{proof}

According to Lemma \ref{l:kfilling} we can solve the $LCSk$
problem using a dynamic programming algorithm working on a two
dimensional table of size $(n-k+1)^2$ where the rows represent the
$A$ sequence and the columns stand for sequence $B$. Cell
$LCSk[i,j]$ contains the value $LCSk_{i,j}$ and the appropriate
predecessors. Nevertheless, when we wish to attain the common
subsequence itself, we encounter a complication.

In the original LCS algorithm, computing the common subsequence,
requires maximizing three options of possible prefixes of the LCS.
When some of these prefixes have the same length, there is no
significance which of them is chosen, as a single common
subsequence is sought  and the selection has no effect on future
matches. However, in the $LCSk$ problem the situation is
different. For example, consider $LCS2$ for the strings of Figure
3.
\begin{figure}[h]\label{f:notLCS}
\begin{center}
A = $\quad $ \begin{tabular}{lllll}
  \footnotesize{1} & \footnotesize{2}  & \footnotesize{3} & \footnotesize{4} & \footnotesize{5}     \\
   G $ \qquad $   & C $ \qquad $  & G   $\qquad  $  & T $ \qquad $  &   C     \\
\end{tabular}\\

$\qquad \  $B = $\quad $ \begin{tabular}{lllll} 
       C  $ \qquad $ & G $ \qquad $   & C  $ \qquad $     & G $ \qquad $    &  T $ \qquad $  \\ 
\end{tabular}
\normalsize \caption{An LCS2 example }
\end{center}
\end{figure}
  LCS2[3,3] equals 1 due to the $2Match$
$(1,2)$ (matching $a_1a_2$ to $b_2,b_3$)  ($G\ C$), or by the
$2Match$ of $(2,1)$ ($C\ G$). In spite of the fact that both
common subsequences share the same length, the former is part of
the final solution as it enables a further $2Match$ at $(3,4)$
while the latter cannot be extended due to the  overlap
restriction. It, therefore, seems that all possibilities of common
subsequence in $k$ length substrings, that is, all predecessors
should  be saved at every calculation in order to enable a correct
backtracking of the optimal solution. As the dynamic programming
proceeds, this information can exponentially increase.
Nevertheless, we prove in Lemma \ref{l:kelimination} that in the
$LCSk$ problem only one maximal previously computed subsequence is
required. The three options of forming $LCSk_{i,j}$, as appear in
Lemma \ref{l:kfilling}, compile $candidates(i,j)$, hence
$pred(i,j)$. Therefore, the $pred(i,j)$ set should be updated
after computing $LCSk_{i,j}$.

\begin{cor}\label{c:candidates union} If \ \ $LCSk_{i,j-1} = LCSk_{i-1,j} =
LCSk_{i-2,j-2}+1$, and kMatch(i,j)=1 then $\ pred(i, j) = pred(i,
j-1)\ \bigcup  pred(i-1, j)\ $
$\bigcup \ (i,j)$.\\

If \ \ $LCSk_{i,j-1} = LCSk_{i-1,j}$  and kMatch(i,j)=0 then\\
$ \ pred(i, j) = pred(i, j-1)\ \bigcup pred(i-1, j)\ $.\\

In both cases, if one or more of the relevant
$LCSk_{x,y},$ \small{$x\leq i, y\leq j $} has shorter length, its
corresponding $pred$ is not included in $pred(i, j)$.
\end{cor}
\begin{proof}
Note, that the length of a predecessor $p\in pred(i,j)$
equals the value of $LCSk_{i,j}$. Due to Lemma \ref{l:length}
there is no necessity to consider the shorter predecessors.
 Suppose all three sets contain predecessors representing
common subsequences of the same length. Without further
information, we cannot determine which common subsequence ending
in $pred(i, j-1),\ pred(i-1, j),$
 or in $k$ matching of $(i,j)$, will be in the maximal output,
therefore, all  predecessors must be considered. \QED
\end{proof}

\subsection{The Backtrack Process}
Using the recursive rule of Lemma \ref{l:kfilling}, the value
computed for $LCSk_{i-k+1, j-k+1}$ is the length of the common
subsequence in $k$ length substrings of sequences $A$ and $B$. In
order to obtain the common subsequence itself we perform the
following procedure. Consider the value saved in cell $LCSk[i,j]$,
where $i$ and $j$ are initialized by $n - k + 1$. We suppose that
each cell contains a single predecessor, as will be proven
hereafter in Lemma \ref{l:kelimination}. Let the predecessor saved
in the current cell be $(x,y)$. Two cases are regarded as long as
$i,j
>0$.
\begin{enumerate}
\item  if $x = i$ and $y = j$, then  a $k$ matching starts in
these indices, therefore $a_{i+f}\ for\ every \ 0 \leq f \leq k-1$
can be added to the constructed output, preserving the increase of
the indices. In order to proceed we decrease both $i, j $ by $k$
to avoid previous $k$ matchings overlapping $(i,j)$.\\
\item Otherwise, no $k$ matching occurs in the current indices.
The predecessor $(x,y)$ directs us to the cell containing a $k$
matching which is part of an LSCk with the value $LCSk_{i,j}$.
Therefore, we decrease the indices $i = x$ and $j = y$.
\end{enumerate}

\subsection{Predecessors Elimination}\label{ss:elimination}
We aim at minimizing the number of predecessors per $LCSk[i,j]$
and therefore define a process of predecessors elimination.
Eliminating a predecessor $p$, that is, deleting it from $pred(i,
j)$ can be safely done if a maximal common subsequence in $k$
length substrings of the same length  can be attained using
another predecessor from $pred(i,j)$. Lemma \ref{l:kelimination}
provides the elimination procedure and its correctness.

\begin{lem}\label{l:kelimination}\textbf{Elimination Lemma.}
Let $p_1$, $p_2  \in pred(i,j)$ be $k\ matching$s, where $|p_1| =
|p_2|$, then one of $p_1, p_2$ can be arbitrarily eliminated.
\end{lem}
\begin{proof}
Let  $p_1= (s,t)$ , $p_2= (s', t')$.
In case $kMatch(i,j) = 0$ then,  if the backtracking  pass through
table cell [i,j]  it implies that the previously found $k$
matching is $(i+k, j+k)$ due to the second case of the
backtracking procedure. Moreover, according to Corollary
\ref{c:candidates union},  both $\{s, s'\} \leq i$ and $\{t , t'\}
\leq j$. As a consequence, there is no preference to one of the
equal length predecessors as both cannot overlap the previous $k$
matching.

Suppose then that
$kMatch(i,j)=1$ and $p_2 = (i,j)$. According to the backtracking
procedure, we get to cell $[i,j]$ either by the first case of the
procedure where there is a $k$ matching $(i+k, j+k)$ or by its
second case where at cell $[i', j']$ there is no $k$ matching but
it contains a predecessor (i,j). The latter implies that the
previously found $k$ matching is $(i + k + f, j + k + h)$ for $f,h
> 0$.

There are two cases to consider.
\begin{enumerate}
\item If no optimal solution uses the $k$ matching $(i,j)$ it
implies that the optimal solution includes  $k$ matchings $(i',
j')$ and $(i'', j'')$ where $i' < i < i'+k$  and $i'' - k < i <
i''$ or $j' < j < j'+k$  and $j'' - k < j < j''$. If only one
inequality holds for $i$ or $j$ then some optimal solution will
include $(i,j)$, contradicting the case definition. According to
the first case of the backtrack procedure, when backtracking from
cell [i'', j''], including the $k$ matching $(i'',j'')$, we
decrease both indices by $k$. Since $i''-k < i$ and $ j''-k < j$
cell [i,j] will not be considered, therefore even if we saved
$p_2$, that is we eliminated $p_1$, it has no
consequence on the optimal solution.\\
\item If there exists an optimal solution including $p_2$ but we
arbitrarily eliminated it.  Since we proved that the previously
found $k$ matching is $(i + k + f, j + k + h)$ for $f,h \geq 0$
there is no preference to $p_2$ over $p_1$ as they are both of the
same length and both do not overlap the previously found $k$
matching according to Corollary \ref{c:candidates union}.
Apparently, $p_1$ is included in another optimal solution.
\end{enumerate}\QED
\end{proof}

The $kMatch$ function requires $k$ symbol matching for each table
entry. Nevertheless, matches can be enlarged only on the diagonal
of the table, from cell $[i, j]$ to cell $[i+1,j+1]$. We therefore
suggest to save at every table cell $ [i,j]$ a diagonal counter,
named $dcount$, counting the length of the longest match between
the suffixes of $A[1 \ldots i]$ and $B[1 \ldots j]$.

While filling the dynamic programming table, at cell $[i,j]$ we
compare merely $a_i$ with $b_j$ and assign $dcount[i,j]$ its value
according to the following definition:
\begin{definition} \label{d:dcount}
\begin{displaymath} dcount[i,j] =
\left\{ \begin{array}{ll}
1+ dcount[i-1,j-1] & \quad if\ a_{i} = b_{j} \\
0 & \quad a_i \neq b_j\\
\end{array} \right.
\end{displaymath}
\end{definition}
Now, instead of using the $kMatch$ function when computing the
score of $LCSk[i,j]$, we need only compare $dcount[i+k-1, j+k-1]$
with $k$. The situation of $dcount[i+k-1, j+k-1]
> k$ occurs when there is an overlap between some matched
substrings, but this is handled by the Elimination Lemma. Note
that we need to proceed with the $dcount$ computations $k$ rows
and columns ahead during the LCSk computation.

\paragraph{\textbf{Example.}}
Figure 3 depicts an LCS2 table. We demonstrate the two cases in
 Lemma \ref{l:kelimination} where $kMatch(i,j)=1$.
For the first case, consider cell $LCS2[5,6]$ including
$pred_{5,6} = \{ (4,5),(5,6) \}$. Suppose we arbitrarily eliminate
$(4,5)$. The LCS2 may contain the 2 matching $(5,6)$ that overlaps
with (4,1),(4,5) and on the same time overlaps also $(6,7)$ what
can decrease the length of the solution. Nevertheless, according
to the backtracking procedure, after considering $LSC2[6,7]$ we
decrease the indices and go to $LCS2[4,5]$ in which $(5,6)$ cannot
exist, due to Corollary \ref{c:candidates union}.

For the second case consider cell $LCS2[2,3]$ including
$pred_{2,3} = \{ (1,1),$ $(2,3)\}$. (2,3) is included in one of
the optimal solution. Suppose we eliminated it and retained
$(1,1)$. The backtracking path goes through cells $[7,7]$ to
$[6,7]$ to $[4,5]$ to $[2,3]$ where it finds a non overlapping
predecessor $(1,1)$ with the same length as the deleted $(2,3)$.

\begin{flushleft}
\begin{figure}[h,t]\label{f:table}
\begin{center}
\footnotesize
\begin{tabular}{|cc|c|c|c|c|c|c|c|c|}
  \hline
 &   &  1     & 2         & 3          & 4         &5         & 6        & 7        &8    \\
 &   &  C    & T          & T          & G         & C        &  T       & T       & T\\ \hline
 &   &   1   & 1          & 1          & 1         &1         & 1        &1        & - \\
1 $\ $ &  C & (1,1) & (1,1)      &(1,1)        & (1,1)    &(1,1)(1,5)&(1,5)     &(1,5)    & - \\
 &   &   dc=1   & dc=0          & dc=0          & dc=0         &dc=1         &dc=0        &dc=0        & dc=0
 \\\hline
&  &    1  & 1          & 1          & 1         &1         & 1        &1        & - \\
2$\ $ &  T & (1,1) & (1,1)      &(1,1)(2,3)  &(1,1) &(1,1)(1,5)
&(1,5)   &  (1,5)& - \\
 &   &   dc=0   & dc=2          & dc=1          & dc=0         &dc=1         &dc=2        &dc=1        &
 dc=1  \\\hline
&  &  1    & 1          &1           & 1         &1         & 1        &1         & - \\
 3$\ $& G  &  (1,1)& (1,1)      &  (1,1)     & (1,1)    &(1,1)(1,5) &(1,5)     &(1,5)     & - \\
  &   &   dc=0   & dc=0          & dc=0          & dc=2         &dc=0         &dc=0        &dc=0        &
 dc=0  \\\hline
&  & 1     &  1         & 1          & 1         &2         & 2        &2         & - \\
 4$\ $& C  & (4,1) &  (4,1)     &(1,1),(4,1) & (4,1)    &(4,5)     &(4,5)     &(4,5) & - \\
   &   &   dc=1   & dc=0          & dc=0          & dc=0         &dc=3         &dc=0        &dc=0        &
 dc=0  \\\hline
&  &  1    & 1          & 1          & 1         & 2        & 2        &2         & - \\
5$\ $ & T  & (4,1) &(4,1),(5,2) &(4,1)       & (4,1)    &(4,5)
&(4,5),(5,6)&(4,5),(5,7)& - \\
   &   &   dc=0   & dc=2          & dc=1          & dc=0         &dc=0         &dc=4        &dc=1        &
 dc=1 \\\hline
&  &  1    &1           & 1          & 1         &2         &2         &3         & - \\
 6$\ $& T  & (4,1) &(4,1),(6,2) &(4,1)       & (4,1)    &(4,5)     &(4,5),(6,6)&(6,7)     & - \\
    &   &   dc=0   & dc=1          & dc=3          & dc=0         &dc=0         &dc=1        &dc=5        &
 dc=2 \\\hline
&  &  1    &  1         &  2         & 2          &2        & 2        & 3        & - \\
 7$\ $&  T & (4,1) &(4,1)       &(7,3)       & (7,3)    &(7,3),(4,5)&(7,3)(6,6)&(6,7)     & - \\
     &   &   dc=0   & dc=1          & dc=2          & dc=0         &dc=0         &dc=0        &dc=2        &
 dc=6 \\\hline
8$\ $& G &  -    &  -        &     -      &    -      &  - & - & -
& - \\
    &   &   dc=0   & dc=0          & dc=0          & dc=3         &dc=0         &dc=0        &dc=0        &
 dc=0 \\\hline
\end{tabular}
\end{center}
\normalsize \caption{An LCS2 Table. The numbers represent the
length of the common subsequence. The pairs in parenthesis stand
for the predecessors. Each cell contains all possible predecessors
according to Corollary \ref{c:candidates union}. Due to the
Elimination Lemma only one predecessor is retained. The diagonal
counter $dcount$ is denoted by dc.}
\end{figure}

\end{flushleft}

\begin{thm}\label{t:k time} The $LCSK(A,B)$ problem can be solved in $O(n^2)$ time and $
O(kn)$ space, where $n$ is the length of the input sequences $A$,
$B$. Backtracking the solution requires time of $O(\ell )$ where
$\ell $ is the number of $k$ matchings in the solution, and
$O(n^2)$ space.
\end{thm}
\begin{proof}
The algorithm fills a table of size $(n - k + 1)^2$. Each entry is
filled according to Lemma \ref{l:kfilling} by performing a
constant number of comparisons in addition to applying
$kMatch(i,j)$. Thanks to the usage of $dcount$, the $KMatch$ is
also reduced to a constant time operations as can be seen at
Definition \ref{d:dcount}.
 In addition, unifying three $pred$ sets of size one each
does not increase the time requirements per entry. The Elimination
procedure requires also constant time according to Lemma
\ref{l:kelimination}.

All in all, constant time operations are
performed for each of the table entries, concluding in $O(n^2)$
time requirement for computing the optimal length of the common
subsequence in $k$ length substrings.

During the backtracking process we go through the cells
representing the $k$ matchings of one optimal solution. If the
difference between two such $k$ matchings is more than $k$, we
go through an intermediate cell whose predecessor directs us
to the next $k$ matching. Hence finding the common subsequence in
$k$-length substrings requires $O(\ell)$ where $\ell$ is the
number of $k$ matchings in the solution.\\

\textbf{Regarding space:} Each of the $n^2$ entries contains,
according to Corollary \ref{c:candidates union} three predecessors
and the Eliminate function, due to Lemma \ref{l:kelimination},
results in a single predecessor before considering further
entries, implying $O(n^2)$ space requirement. Nevertheless, due to
Lemma \ref{l:kfilling}, during the computation of $LCSk[i,j]$ we
need only row  $i-k$ and column $j-k$ as well as cell
$[i+k-1,j-k+1]$ for its $dcount$ value. As a consequence, at each
step we save  $O(k)$ rows and columns implying the space
requirement is $O(kn)$. In order to backtrack the solution, the
whole table is needed, implying $O(n^2)$ space requirement. \QED
\end{proof}

\section{Edit Distance in k Length Substrings}\label{s:editdistance}
The well-known edit distance suggested by Levenstein \cite{l:66} is strongly
 related to the LCS similarity measure. The problem is formally defined as follows:\\

\begin{definition}\label{d:ed}
{\em The Edit Distance ($ED$) Problem}:\\
\begin{tabular}{ll}
 Input: &  Two sequences $A = a_1a_2\ldots a_n,$  $B=b_1b_2\ldots b_m$ over alphabet
$\Sigma$.\\
  Output: & The minimal number of  insertions, deletions
  and substitutions required to\\ & transform A into B.\\
\end{tabular}
\end{definition}

The Edit Distance problem is considered as a complement to the LCS problem
in the following sense. When the allowed edit operations are
insertions and deletions we have:
\[ED(A,B) = |A| + |B| - 2LCS(A,B).\]
Having this relation to the LCS problem in mind, we would like to define an
edit distance with regard to $k$-length substrings that will have a similar complementary nature
to the LCS with $k$-length substrings. Definition~\ref{d:EDk} seems natural.

\begin{definition}\label{d:EDk}
   \emph{The Edit Distance in k Length Substrings
    Problem ($EDk$)}:\\
\begin{tabular}{ll}
 Input: &  Two  sequences $A= a_1a_2 \ldots a_n$,  $B=b_1b_2 \ldots b_m $  over alphabet
$\Sigma$. \\
   Output: & The minimal number of insertions, deletions and substitutions required\\
   & to transform $A$ into $B$, while symbols between locations on which edit\\
   & operations occur are $k$-length substrings forming $LCSk$ of $A$ and $B$.\\
\end{tabular}
\end{definition}

Note that according to Definition~\ref{d:EDk}, it is possible that between two locations on which an edit operation occur there is a
common substring of length greater than $k$ but less than $2k$.
Thus, this definition enables considering only $k$ consecutive symbols as a
common $k$-length substring. The rest of the common symbols are
counted as substitution edit operations. This is due to the fact that
 in the edit distance problem there are matched symbols and unmatched symbols.
 However, the $EDk$ distance of Definition~\ref{d:EDk} considers only  not overlapping
 $k$ matchings,   therefore, any common substring
shorter than $k$ symbols cannot be considered matched.

In order to capture this difference, consider for example the sequences of Figure \ref{f:pair example}. For these sequences
$LCS2(A,B) = 2$ and a suggested common subsequence is $G\ T\ T\ G$
given by $2Match(4,1)$ and $2Match(7,5)$. The prefix $A[1..3]$ has 3 insertions
that occur before the first $2Match$. Between both $2Match$es, sequence
$A$ has a single symbol $G$ and sequence $B$ has two symbols $T\ G$.
Although $G$ appears in both sequences, it is considered as a
substitution error due to the fact that $k = 2$, forcing a
substitution and deletion errors. After the last $2Match$, $B$ has
additional two symbols, implying two insertion errors in $A$. All
in all, we have $EDk(A,B) = 7$.

Such a definition may be more accurate than the known edit distance when
we do not allow changes to occur at every odd location for example. Instead,
we would like them to be distributed on the inputs, forcing some
 $k$-length substrings be identical allowing no edit operation in these substrings.
 The $EDk$ is related to $LCSk$ as the $ED$ is related to $LCS$. This is stated
 in the following observation.
\commentout{
\begin{obsr}\label{o:EDbyLCS}
 Let a longest common
subsequence of $n$-length string $A$ and $m$-length string $B$ be
$x_1x_2\ldots x_\ell $ and let $i_1,i_2,\ldots i_\ell $ and
$j_1,j_2,\ldots j_\ell $ be the locations of $x_1x_2\ldots x_\ell
$ in the input strings $A$ and $B$, respectively. I.e.,\
$a_{i_f}=x_f=b_{j_f}$, where $0 < f \leq \ell$. Assume that
$i_{\ell + 1}= n$, $j_{\ell + 1} = m$. Then,
\[EDk(A , B) = \max \{i_1,j_1\} + \Sigma_{f=1}^{\ell } \max\{ (i_{f+1} - i_f - k), (j_{f+1} - j_f - k)\}.\]
\end{obsr}
}
\begin{obsr}\label{o:EDbyLCS}
In case the allowed edit operations are insertions and deletions,
\[EDk(A , B) = |A|+|B|-2k\cdot LCk(A, B).\]
\end{obsr}

In spite of Observation \ref{o:EDbyLCS}, the $EDk$ cannot be easily
 found based on finding an $LCSk$ solution, because the substitution
 operation is usually used so different longest common subsequences
  in $k$-length substrings may define different number of required
  edit operations to transform
one input into the other. To see this, consider the $LCS2$ for the
sequences appearing in Figure \ref{f:pair example}. The maximal
number of $2$ matchings between the inputs is two, however there
are several options to select these $2$ matchings, two of which
appear in Figure \ref{f:2waysED}.
\begin{figure}[h]
\begin{center}
A =  \begin{tabular}{llllllll}
       \footnotesize{1} & \footnotesize{2}  & \footnotesize{3} & \footnotesize{4} &\footnotesize{5} & \footnotesize{6} & \footnotesize{7} &\footnotesize{8}    \\
      T$ \quad $    & G$ \quad $   & C$ \quad $   & G$ \quad $    & \textbf{T}$ \quad $     &  \textbf{G}$ \quad $     & \textbf{T}$ \quad $    & \textbf{G}\\
\end{tabular}

B =  \begin{tabular}{llllllll} 
      G$ \quad $ & T$ \quad $ & \textbf{T}$ \quad $    & \textbf{G}$ \quad $      & \textbf{T}$ \quad $    &  \textbf{G}$ \quad $      & C$ \quad $   & C 
\end{tabular}\\
\vspace{0.5cm}
 A =  \begin{tabular}{llllllll}
       \footnotesize{1} & \footnotesize{2}  & \footnotesize{3} & \footnotesize{4} &\footnotesize{5} & \footnotesize{6} & \footnotesize{7} &\footnotesize{8}    \\
      \textbf{T}$ \quad $    & \textbf{G}$ \quad $   & C$ \quad $   & G$ \quad $    & T$ \quad $     &  G$ \quad $     & \textbf{T}$ \quad $    & \textbf{G}\\
\end{tabular}

B =  \begin{tabular}{llllllll} 
      G$ \quad $ & T$ \quad $ & \textbf{T} $ \quad $    & \textbf{G}$ \quad $      & \textbf{T}$ \quad $    &  \textbf{G}$ \quad $      & C$ \quad $   & C 
\end{tabular}
\normalsize \caption{2 options for LCSk and the EDk
derived}\label{f:2waysED}
\end{center}
\end{figure}
The upper option refers to $2$ matchings$(5,3)$ and $(7,5)$
implying edit distance of 6, including 2 substitution operations.
The lower option of considering $(1,3)$ and $(7,5)$ yielding 8
required edit operations. Consequently, finding a specific $LCSk$
solution does not guarantee an optimal sequence of edit operations
giving the optimal $EDk$ distance and a separate dynamic
programming for the $EDk$ is in order.

\subsection{Solving the EDk Problem}\label{ss:EDkalgorithm}
We suggest solving the $EDk$ problem using dynamic programming,
similarly to the original Edit Distance problem, with the required
modifications. Entry $EDk[i,j]$ contains the score of the edit
distance in $k$ length substrings of the prefixes $A[1...i]$ and
$B[1...j]$. The score of $EDk[i,j]$ minimizes the number of edit
operations required so far, taking into account all three edit
operations. If we encounter $a_i $ not matching $ b_j$, it means a
required insertion or deletion of $a_i$ or a substitution
operation. In the other case of $a_i = b_j$, we need to verify
that a $k$ matching ends at $(i,j)$ in order to declare a $k$
length common substring requiring no edit operation. This can be
easily ascertained by comparing $dcount[i,j]$ with $k$.

According to  Definition \ref{d:EDk} unedited symbols should form a legal $LCSk$. Hence, it
is necessary to verify that no two overlapping $k$ matchings are left unedited. Naturally, $EDk[i,0] = i$ and $EDk[0,j] = j$. Lemma
\ref{l:EDkfilling} below formally describes the computation of
$EDk[i,j]$.

\begin{lem}{\textbf{The EDk Recursive Rule.}}\label{l:EDkfilling}
\begin{displaymath} EDk[i,j] = min \left\{
\begin{array}{l}
EDk[i-1,j] + 1  \\
EDk[i,j-1] + 1   \\
\left\{
\begin{array}{ll}
EDk[i-k,j-k]  & \quad {\rm if}\ dcount[i,j]\geq k  \\
EDk[i-1,j-1] + 1 & \quad {\rm if}\ dcount[i,j]< k  \\

\end{array} \right.\\
\end{array} \right.\\
\end{displaymath}
\end{lem}
\begin{proof}
$EDk[i,j]$ contains the minimal number of edit operations required
in order to transform sequence $A$ into sequence $B$, with respect
to $k$ length substrings, preserving their order in the input
sequences. The $EDk[i,j]$ is obtained by minimizing the score due
to all possible edit operations:
\begin{enumerate}
\item Considering $a_i$ as a deleted symbol, implies increment of
$EDk[i-1, j]$.
 \item Considering $b_j$ as a redundant symbol,
hence an insertion in the $A$ sequence is in order, implies
increment of $EDk[i, j-1]$.
 \item Considering $a_i$ and $b_j$ as the last matching symbols of a $k$ matching or
 part of a substitution. Note that the same case holds for  $a_i \neq b_j$ and $a_i = b_j$ if the
 equal symbols are part of a different $k$ matching or a shorter matching due to the problem and
 and the table entries definitions.
  We distinguish between these options by
 comparing  $dcount[i, j]$ with $ k $.
\begin{enumerate}
  \item The case of $a_i$ and $b_j$ are the last matching symbols of a $k$ matching
   implies $dcount[i,j] \geq k$ as $(i-k+1,j-k+1)$ is a common
 $k$ length substring. Therefore, no increase to the previously
 computed score is required. Nevertheless, due to the necessity to avoid
 overlaps between $k$ matchings, when $(i-k+1,j-k+1)$ is a $k$
 matching, the $EDk$ computation must verify that previous score does not
 take into account an overlapping $k$ matching.  Therefore, we
 consider $EDk[i-k,j-k]$ referring  to the minimal number of required  edit operations
 when the last $k$ matching can end  by matching at most $a_{i-k}$ with $
 b_{j-k}$.
 \item For the case of $dcount[i, j]< k$, as there is no $k$
 matching ending at $a_i$ and $b_j$, we have a single substitution error.
 Consequently, there is no restriction on the ending of a previous $k$
 matching, as no overlap can occur at these indices, so we
   increment the score of $EDk[i-1,j-1]$.
 \end{enumerate}
\end{enumerate}
These claims can be easily proven by induction.

Note that as every symbol of the inputs
is either part of a $k$ matching or it represents an edit error,
as a consequence, the minimization of the $EDk$ scores avoiding
overlaps between $k$ matchings also maximizes the number of $k$
length common substrings.\QED
\end{proof}

\paragraph{\textbf{Example.}}
Figure \ref{f:EDktable} depicts an ED2 table. We demonstrate all
possible cases in the process of obtaining a minimal score.
\begin{enumerate}
\item Considering $ED2[6,2]$, the minimal value that can be
obtained is  4, by incrementing  $ED2[5,2]$
 \item  Considering
$ED2[4,7]$, the minimal value that can be obtained is
 3, by incrementing  $ED2[4,6]$
\item  \begin{enumerate}
  \item  Considering $ED2[4,5]$, we see $dcount[4,5] = 3 > 2$ implying a 2
  matching $(3,4)$ (GC) ends by matching $a_4$ with $b_5$. In order to avoid overlaps, we use the
  score  of  $ED2[2,3]$. This value is smaller than the other possible values in this entry.
  \item Considering $ED2[3,3]$, we see $dcount[3,3] = 0 $ implying no 2
  matching  ends by matching $a_3$ with $b_3$. With no overlaps to avoid, we increment the
  score  of  $ED2[2,2]$. The obtained value is smaller than the other possible values in this entry.
  \end{enumerate}
\end{enumerate}

\begin{flushleft}
\begin{figure}[h]\label{f:EDktable}
\begin{center}
\footnotesize
\begin{tabular}{|cc|c|c|c|c|c|c|c|c|c|}
  \hline
 & & 0  &  1     & 2         & 3          & 4         &5         & 6        & 7        &8    \\
 &  &  &  C    & T          & T          & G         & C        &  T       & T       & T\\ \hline
0 $\ $& &  0 &   1   & 2          & 3          & 4         &5 & 6
&7 & 8 \\\hline
 1 $\ $&  C & $\ $1$\ $&   1   & 2          & 3          & 4         &5         & 6        &7        & 8 \\
  &   & &   $\ $dc=1   & $\ $dc=0          & $\ $dc=0          & $\ $dc=0         &$\ $dc=1         &$\ $dc=0        &$\ $dc=0        & $\ $dc=0
 \\\hline
2$\ $& T & $\ $2$\ $&    2  & 0          & 1          & 2        &3         & 4        &5        & 6 \\
  &   & & $\ $ dc=0   & $\ $dc=2          &$\ $dc=1          &$\ $dc=0         &$\ $dc=1         &$\ $dc=2        &$\ $dc=1        &
 dc=1  \\\hline
 3$\ $& G & $\ $3$\ $&  3    & 1          &1           & 2         &3         & 4        &5         & 6 \\
  &   & &  $\ $dc=0   &$\ $dc=0          &$\ $dc=0          &$\ $ dc=2         &$\ $dc=0         &$\ $dc=0        &$\ $dc=0        &
 dc=0  \\\hline
4$\ $& C & $\ $4$\ $& 4     &  2         & 2          & 2         &1         & 2        &3         & 4 \\
    &  &  &  $\ $dc=1   & $\ $dc=0          & $\ $dc=0          &$\ $dc=0         &$\ $dc=3         &$\ $dc=0        &$\ $dc=0        &
 dc=0  \\\hline
5$\ $& T & $\ $5$\ $&  5    & 3          & 3          & 3         & 2        & 2        &3         & 4 \\
    &   & &   $\ $dc=0   & $\ $dc=2          & $\ $dc=1          & $\ $dc=0         &$\ $dc=0         &$\ $dc=4        &$\ $dc=1
    &  dc=1 \\\hline
6$\ $& T & $\ $6$\ $&  6    &4           & 4          & 4         &3         &3         &1         & 2 \\
     &   & &  $\ $ dc=0   & $\ $dc=1          &$\ $dc=3          & $\ $dc=0         &$\ $dc=0         &$\ $dc=1        &$\ $dc=5        &
 dc=2 \\\hline
7$\ $& T & $\ $7$\ $ &  7    &  5         &  5         & 5          &4        & 4        & 2        & 2\\
     &  &  &  $\ $ dc=0   & $\ $dc=1          & $\ $dc=2          & $\ $dc=0         &$\ $dc=0         &$\ $dc=0        &$\ $dc=2        &
 dc=6 \\\hline
8$\ $& G & $\ $8$\ $ &  8    &  6        &     6      &    4 &  5  & 5 & 3  & 3 \\
    &   & &  $\ $ dc=0   & $\ $dc=0          &$\ $ dc=0          &$\ $dc=3         &$\ $dc=0         &$\ $dc=0        &$\ $dc=0        &
 dc=0 \\\hline
\end{tabular}
\end{center}
\normalsize \caption{An ED2 Table. The numbers represent the
minimal number of required edit operations. The diagonal counter
$dcount$ is denoted by dc.}
\end{figure}
\end{flushleft}

\begin{thm}\label{t:EDk time} The $EDk(A,B)$ problem can be solved in $O(nm)$ time and $
O(km)$ space, where $m$, $n$ are the lengths of the input
sequences $A$ and $B$ respectively. Backtracking the solution
requires time of $O( EDk(A,B) +\ell )$ where $EDk(A,B)$ is the
optimal score and  $\ell$ is the number of $k$ matchings
considered in the solution, and $O(nm)$ space.
\end{thm}
\begin{proof}
Similarly to the proof of Theorem \ref{t:k time}, we fill a table
of size $nm$. Each entry is filled according to Lemma
\ref{l:EDkfilling} by performing a constant number of comparisons.
The computation of the diagonal count $dcount[i,j]$ requires a
constant time as well as appears in Definition \ref{d:dcount}.
Hence,  $O(nm)$ time  is required for computing the optimal number
of edit operations with respect to the longest  common subsequence
in $k$ length substrings.

In case we want to find  selected edit operations and a common
subsequence in $k$ length substrings suitable to the optimal $EDk$
score, we need to save at each  table entry the indices of the
entry from which the minimal score was deduced. Having this
information, we start the backtracking from $EDk[n,m]$ and report
an operation according to the indices saves there: $[n-1,m]$ imply
a deletion, $[n,m-1]$ imply insertion, $[n-1,m-1]$ refer to
substitution and $[n-k,m-k]$ stand for a $k$ matching. We repeat
the process with the $EDk[f,g]$ where $[f,g]$ are the indices
saved at the previous entry until $f=0 $ or $g =0$. All in all,
every edit operation and $k$ matching requires a constant time
backtrack, yielding $O(EDk(A,B) +\ell )$ time requirement for the
backtracking process. Note that this time is bounded by the
maximal number of edit operation which is $\max \{n,m\}$.

\textbf{Regarding space:} Each of the $nm$ entries contains a
constant number of values. Nevertheless, due to Lemma
\ref{l:EDkfilling}, during the computation of $EDk[i,j]$ we need
only $k$ rows    backwards. As a consequence, at each step we save
$O(k)$ rows  implying the space requirement is $O(km)$. In order
to backtrack the solution, the whole table is needed, yielding
$O(nm)$ space requirement. \QED
\end{proof}

\section{Conclusion}\label{s:conclusion}
In this paper we defined a generalization of the LCS problem,
where each matching must consist of $k$ consecutive symbols, and
by thoroughly understanding the traits of the problem proved an
algorithm with the same time complexity as the special case of LCS
can solve the generalized problem. We also considered an adequate
complementary edit distance measure and showed similar results
hold also for this distance measure. As we consider the LCSk as a
more accurate sequence similarity measure, we believe this problem
should also be studied for generalized sequences, such as weighted
sequences \cite{ags:10}. Other ways for changing the traditional
LCS definition to obtain more accurate similarity measures may
also be suggested.

\end{document}